\documentclass[12pt]{article}
\usepackage{graphicx}
\usepackage{amsfonts,amsmath}
\usepackage[mathscr]{eucal}
\usepackage{amssymb}
\usepackage{amsthm}
\usepackage{bbold}
\theoremstyle{plain}
\newtheorem{thm}{Theorem}

\newcommand{\be}{\begin{eqnarray}}
\newcommand{\ee}{\end{eqnarray}}
\newcommand{\bc}{\begin{center}}
\newcommand{\ec}{\end{center}}
\newcommand{\nn}{\nonumber \\}
\newcommand{\lb}{\label}
\newcommand{\p}[1]{(\ref{#1})}

\begin{document}

\begin{titlepage}

\vspace*{0.2cm}

\renewcommand{\thefootnote}{\star}
\begin{center}

{\LARGE\bf  $Spin(7)$  and generalized $SO(8)$ instantons in eight dimensions.}

\vspace{2cm}

{\Large A.V. Smilga} \\

\vspace{0.5cm}

{\it SUBATECH, Universit\'e de
Nantes,  4 rue Alfred Kastler, BP 20722, Nantes  44307, France. }

\end{center}
\vspace{0.2cm} \vskip 0.6truecm \nopagebreak

   \begin{abstract}
\noindent  

We present a simple compact formula for a topologically nontrivial map $S^7 \to Spin(7)$ associated with the fiber bundle $Spin(7) \stackrel{G_2}{\to} S^7$. The homotopy group  $\pi_7[Spin(7)] = \mathbb{Z}$ brings about the 
topologically nontrivial 8-dimensional gauge field configurations that belong to the algebra 
$spin(7)$. The instantons are  special such configurations that minimize the  functional $\int {\rm Tr} \{F\wedge F \wedge \star(F \wedge F)\} $ and satisfy non-linear self-duality conditions, $ F \wedge F \ =\ \pm  \star (F\wedge F)$. 

$Spin(7) \subset SO(8)$, and $Spin(7)$  instantons represent simultaneously $SO(8)$ instantons of a new type. The relevant homotopy is $\pi_7[SO(8)] = \mathbb{Z} \times \mathbb{Z}$, which implies the existence of two different topological charges. This also holds for all groups $SO(4n)$ with integer $n$. We present explicit expressions for two topological charges and calculate their values for the conventional 4-dimensional and 8-dimensional instantons and also for the 8-dimensional instantons of the new type.
  
Similar constructions for other algebras in different dimensions are briefly discussed. 

   \end{abstract}

\end{titlepage}

\setcounter{footnote}{0}

\setcounter{equation}0

\section{Introduction and Summary}

Eight-dimensional instantons of two varieries have been discussed in the literature. There are the so-called octonionic or secular instantons \cite{FubNic} that satisfy certain modified linear self-duality conditions (these instantons are outside the scope of this paper) and there are $SO(8)$ instantons where the gauge fields satisfy the {\it nonlinear} self-duality conditions

        \be
       \lb{8-self-dual}
       F \wedge F \ =\ \pm  \star (F\wedge F)\,.
        \ee
   $SO(8)$ instantons \cite{Keph,Tchr} are associated with topologically nontrivial mappings $S^7 \to SO(8)$.  In Refs. \cite{Keph,Tchr}, the simplest  such mapping was discussed:
    \be
   \lb{g-Spin8}
  g_8 \ =\ \exp\{i \alpha_m \Gamma_m\}  = \mathbb{1} \cos \alpha + i \frac {\sin \alpha}\alpha \alpha_m \Gamma_m\,,
   \ee
   where $\alpha = \|\alpha_m \| \leq \pi$ and $\Gamma_{m = 1,\ldots, 7}$ are purely imaginary antisymmetric 8-dimensional matrices satisfying
   the Clifford algebra,
          \be 
          \lb{Cliff7} 
          \Gamma_m \Gamma_n + \Gamma_n \Gamma_m \ =\ 2 \delta_{mn} \mathbb{1} \,,
           \ee
           and the relation
            \be
           \lb{prod-gamma}\Gamma_1 \Gamma_2 \Gamma_3 \Gamma_4\Gamma_5\Gamma_6\Gamma_7 \ =\  i \mathbb{1} \, .
              \ee
   The mapping $g_8$ is associated with the instanton satisfying $F \wedge F \ =\ \star (F\wedge F)$ and the mapping $g_8^{-1}$ with the antiinstanton satisfying $F \wedge F \ =\ -\star (F\wedge F)$. The explicit expressions for these instantons are quite analogous to the familiar 4-dimensional BPST instantons \cite{BPST}.   Multiinstanton solutions, the analogs of ADHM solutions \cite{ADHM}, are also known \cite{ADHM-8}. 

The $SO(8)$ instantons satisfying the nonlinear self-duality conditions \p{8-self-dual} have several  interesting physical applications. They were discussed in the string theory context \cite{string} and in the context of the 8-dimensional quantum Hall effect \cite{Hall}. 
   
  All these configurations are characterized by a nonzero integer values of the fourth Chern class \footnote{A ``physical" way to fix the coefficient
 in \p{Chern4} is to require that the Dirac equation on the gauge field background $A_M(x)$  is well defined and has a single 
 zero mode. For the Chern class of order $D/2$ in $R^D$, one derives \cite{kniga}
 \be
 \lb{Chern-D}
 q \ =\ \pm \frac {1}{(2\pi)^{D/2} (D/2)!} \int_{R^D} {\rm Tr} \, \{\overbrace{F \wedge \cdots \wedge F}^{D/2}\}
 \ee
with the sign depending on the convention.}
 \be
 \lb{Chern4}
 q \ =\ \frac 1{384 \pi^4} \int_{R^8} {\rm Tr} \, \{F \wedge F \wedge F \wedge F\}\,.
  \ee
 With this convention, $q=1$ for the instanton and $q = -1$ for the antiinstanton.
 The construction of Ref. \cite{Keph} is quite parallel to the construction of the ordinary 4-dimensional instantons, which are associated with the  mappings $S^3 \to SU(2)$ and have an integer topological charge (second Chern class or Pontryagin number),
  \be
 \lb{Pont}
 q \ =\ \frac 1{8\pi^2} \int_{R^4} {\rm Tr} \, \{F \wedge F \}\,. 
  \ee  
  However, for $SO(8)$ gauge configurations in 8 dimensions, the topology is more complicated than in the BPST case. The relevant homotopy group
  is \cite{Mimura}
     \be
     \lb{ZZ}
   \pi_7[SO(8)] = \mathbb{Z} \times \mathbb{Z}\,,
   \ee
    so that a generic gauge field configuration is characterized by  {\it two} different topological charges. 
  We will show later that the second topological charge is expressed as \cite{Nikita}
    \be
     \lb{Gauss8}
     \tilde {q} \ =\  \frac 1{3 \cdot 2^{11} \cdot \pi^4 } \varepsilon^{a_1  \cdots  a_8} \int F^{a_1 a_2} \wedge F^{a_3 a_4} \wedge F^{a_5 a_6}  \wedge F^{a_7 \,a_8} \,.
      \ee 
    This construction can be generalized to $D=4n$ dimensions with any $n$.  The presence of the second topological charge for the $4n$-dimensional gauge fields belonging to the algebra $so(4n)$ can be understood by tracing the analogy with $4n$-dimensional gravity. The charge \p{Gauss8} has the same structure (and the same coefficient!) as the Gauss-Bonnet invariant, giving the Euler characteristic of an even-dimensional manifold. For $D = 4n$, 
     
      \be
  \lb{Gauss-n}
  \chi \ =\ \frac 1{(2n)! \pi^{2n} 2^{4n} } \varepsilon^{A_1 A_2 \cdots A_{4n-1} A_{4n}} \int  \, R^{A_1 A_2} \wedge \cdots \wedge  R^{A_{4n-1} \,A_{4n}} \,, 
   \ee 
   where $R^{AB} = \frac 12\, R^{AB}_{\ \ MN} \,dx^M \wedge dx^N$ is the curvature 2-form --- the field density for the spin connection form.
   
   The integral representation for the Euler characteristic, the analog of \p{Gauss-n}, can be written for any even-dimensional manifold, also for $D = 4n+2$. Similarly, the topological charge of the type \p{Gauss8} can also be defined for $SO(4n+2)$ gauge fields in $4n+2$ dimensions. However, the Chern class of order $4n+2$, the analog of the integral \p{Chern-D} with an odd number of $F$ factors, {\it vanishes} in this case due to the skew symmetry of $F$.\footnote{One can notice at this point that, while the topological charge \p{Gauss8} has the same structure as the Gauss-Bonnet invariant, the Chern classes \p{Pont} and \p{Chern4} have the same structure as  the {\it Hirzebruch signature}  invariant,
    \be
    \lb{Hirzebruch}
    \tau \ \propto\ \int {\rm Tr} \{R \wedge \ldots \wedge R\}\,,
     \ee
     which is nonzero only for $D = 4n$.}
     This matches well with the topological fact that $\pi_{D-1}[SO(D)]$ has two $\mathbb Z$  factors only for $D = 4n$.
   
   \vspace{1mm}
   
   The expression \p{g-Spin8} gives a mapping of $S^7$ in $SO(8)$.
   Hovewer,  it will be more convenient for us to treat $g_8$ as a matrix that rotates 8-dimensional 8-component real spinors. The set of all such matrices form the so-called $Semispin(8)$ group \cite{semispin}. By triality, this group is isomorphic to $SO(8)$ and has the center $Z_2$ (the matrices $g = \mathbb{1}$ and  $g = -\mathbb{1}$). The double cover of $Semispin(8)$ is $Spin(8)$ with the center $Z_2 \times Z_2$. 
   
   $Semispin(8) \simeq SO(8)$ has a $Spin(7)$ subgroup, and one can be interested in the tolopogically nontrivial mappings $S^7 \to Spin(7)$, which exist due to nontrivial 
   $\pi_7[Spin(7)] \ = \ \mathbb{Z}$, and the associated 8-dimensional instantons with $Spin(7)$  gauge group. 
We may observe at this point that the mapping \p{g-Spin8} {\it is} not such mapping --- the commutators of the generators $\Gamma_m$ in \p{g-Spin8} give 21 generators of $SO(7)$ and, together with $\Gamma_m$, this completes the full $so(8)$ algebra.
   
  As we will see later, the simplest map $S^7 \to Spin(7)$ has a somewhat more complicated  than \p{g-Spin8} form:
    \be
  \lb{g-Spin7}
  g \ =\ \exp \left\{- \frac 12 \alpha_m f_{mnk} \Gamma_n \Gamma_k \right\}\,,
   \ee
   where $f_{mnk}$ are the structure constants of the octonion algebra and $\alpha_m$ is a 7-dimensional vector with a norm not exceeding $\pi$.
   
   We will also see that the map \p{g-Spin7} represents a local section of the fiber bundle $Spin(7) \stackrel{G_2}{\to} S^7$. 
   This map gives rise to the associated $Spin(7)$ instantons, which, evidently, are also $SO(8)$ instantons.
   
   The  expressions  \p{Chern4} and \p{Gauss8} for the topological charges are written for gauge field configurations. There are associated expressions for the topological charges describing the maps   $S^7 \to SO(8)$. They read
   
    \be 
  \lb{q7}
  q \ =\ - \frac 1{2^7 \cdot 3 \cdot 35 \, \pi^4}  \int_{S^7} {\rm Tr} \, \{(g^{-1} dg)^7\}\,.  
  \ee
  and
  \be
  \lb{Gauss-S7}
  \tilde q \ = \ - \frac i{3\cdot 2^{11} \cdot 35 \pi^4} \, \varepsilon^{abcdefgh} \int_{S^7} 
   A^{ab} \wedge (A\wedge A)^{cd} \wedge (A\wedge A)^{ef} \wedge  (A\wedge A)^{gh}\,.
    \ee
    
    We will calculate these charges for the maps \p{g-Spin8} and \p{g-Spin7} and also for the maps $g =  g_7 g_8$ and $g = g_7 g_8^{-1}$.
    The results are quoted in  the table below. The charge $q$ for both $g_8$ and $g_7$ is  unity [this is so due to our choice of the sign in \p{g-Spin7}]. The topological charges of $g^{-1}$ are opposite to the charges of $g$. The charges for the map $g_8 g_7$ is the same as for $g_7 g_8$, etc. It is interesting that the charge $\tilde q$ takes the value $\tilde q =1$ for the map $g_8$, but it is negative for the map $g_7$, involving the extra factor 3.

    \vspace{.5cm}
\bc

\begin{tabular}{l||c|c}
\lb{table8}
  map & $q$ & $\tilde q$ \\
\hline
\hline
$g_8$ & 1 & 1 \\
\hline
$g_7$ & 1 & -3 \\
\hline
$g_7 g_8$ & 2 & -2 \\
\hline 
$g_7 g_8^{-1}$ & 0 & -4 \\
\end{tabular}

\ec
\vspace{.4cm}

{\bf Table 1}. Topological charges of the maps $S^7 \to SO(8)$ and of the associated $SO(8)$ instantons.
  
   \vspace{2mm}

  The topological charges for a generic composite mapping $g = g_8^n g_7^m$ with integer $n,m$ are
\be
q(g_8^n g_7^m) \ =\ n+m, \qquad \tilde q(g_8^n g_7^m) \ =\ n - 3m\,.
\ee

   The plan of the subsequent discussion is the following. In the next section, we first remind the well-known BPST construction for the 4-dimensional $SU(2)$ instantons and then discuss the second topological charge for the $SO(4)$ gauge configurations and present the simple expressions for two types of instantons.  
   
In Sect. 3, we describe the conventional $SO(8)$ instantons associated with the mapping \p{g-Spin8}.

   The heart of the paper is Sect. 4, where we discuss in detail the instanton constructions for $Spin(7)$ and generalized instantons in $SO(8)$.
   
   Sect. 5 is devoted to a brief discussion of
    topologically nontrivial gauge field configurations  in other dimensions with other gauge groups. For $D = 4n$ with any $n$, there exist two types of instantons for the groups $SO(4n)$ and only one such type for the groups $SO(4n-1)$. The corresponding homotopies are  $\pi_{4n-1}[SO(4n)] = \mathbb{Z} \times \mathbb{Z}$ 
    and  $\pi_{4n-1}[Spin(4n-1)] = \mathbb{Z}$.
    
    Topological nontrivial configurations exist also in other cases. We will discuss in particular  6-dimensional fields with $SU(3)$ gauge group. Their nontrivial topology is associated with $\pi_5[SU(3)] = \mathbb{Z}$ and with the fiber bundle 
    $SU(3) \stackrel{SU(2)}{\longrightarrow} S^5$. We present the known explicit expression for the topologically nontrivial map $S^5 \to SU(3)$.

\section{Four dimensions}
 \setcounter{equation}0
\subsection{$SU(2)$}

Consider a 4-dimensional Euclidean gauge field $ A(x) = A_\mu(x) dx_\mu \in su(2)$. Assume that   the field density form
\be 
F(x) = dA(x) - i A(x) \wedge  A(x) \ \equiv \ \frac 12 F_{\mu\nu}\, dx_\mu \wedge dx_\nu\,,
 \ee
 where $F_{\mu\nu} = \partial_\mu A_\nu - \partial_\nu A_\mu - i[A_\mu, A_\nu]$,
 decays fast enough at $x \to \infty$, so that the integral \p{Pont} for the  Pontryagin index
  converges. As is very well known, \p{Pont} is integer in this case. This integer is associated with the nontrivial 
  $\pi_3[SU(2)] = \mathbb{Z}$. 
  
  \vspace{1mm}
  
  Indeed, the form Tr$\{F \wedge F\}$ is exact:
  
  \be 
  \lb{dom3}
 {\rm Tr} \{ F \wedge F\} \ =\ d\omega_3, \qquad {\rm where} \qquad \omega_3 \ =\ {\rm Tr}  \left\{ A \wedge F  +
  \frac {i}3 A \wedge A \wedge A \right\}\,.
  \ee
  Then the integral \p{Pont} boils down to the surface term,
   \be
   q \ =\ \frac 1{8\pi^2} \int_{S^3}  \omega_3 \,, 
    \ee
where $S^3$ is a 3-sphere at large distance from the origin. Bearing in mind that $F$ rapidly decays at large $x$, the potential $A$ acquires at large distances a pure gauge form\footnote{See Ref.  \cite{Uhlen} for a mathematical proof.} 
 \be
 \lb{pure-gen}
 A \ =\ ig^{-1} d g
  \ee
  with $g \in SU(2)$. We derive
  \be
  \lb{q3}
  q \ =\ \frac 1{24\pi^2} \int_{S^3} {\rm Tr} \, \{g^{-1} dg\, g^{-1} dg\, g^{-1} dg \}\,.
  \ee
The simplest nontrivial mapping $S^3 \to SU(2)$ may be presented as 
 \be
 \lb{g-krugl}
 g \ =\ \frac {\mathbb{1} x_4  - i x_m \sigma_m}{r}\,,
 \ee
 where $\mathbb{1}$ is the unity matrix, $\sigma_{m = 1,2,3}$ are the Pauli matrices and $x_\mu = ( x_m, x_4)$ marks a point on the distant sphere of radius $r = \sqrt{x_\mu x_\mu }$. The gauge  field \p{pure-gen} acquires the form
  \be
  \lb{pure-krugl}
  A_\mu \ =\ \frac {\eta^a_{\mu\nu} x_\nu \,\sigma^a }{r^2} \,,
   \ee
   where $\eta^a_{\mu\nu}$ is the `t Hooft symbol \cite{Hooft-Ans},
    \be
    \eta^a_{m4} \ =\ -  \eta^a_{4m} \ =\ \delta_{am}, \qquad   \eta^a_{mn} = \varepsilon_{amn} \,.
    \ee
 To calculate the integral \p{q3}, we represent it in a little bit more explicit form
 \be
  \lb{q-krugl}
  q \ =\ \frac 1{24\pi^2}  \int_{S^3}   \varepsilon_{\mu\nu\rho\sigma}\, \frac {x_\mu}r \,{\rm Tr} \, \{g^{-1} \partial_\nu g\, g^{-1} \partial_\rho g\, g^{-1} \partial_\sigma g \}\,\,  dV_{S^3}
  \ee
  with the convention $\varepsilon_{1234} = 1$.
 Pick up the north pole of $S^3$, $x_\mu = (\vec{0}, r)$. Then the integrand in 
 \p{q-krugl} reads
 \be
 \lb{integrand}
 - \frac i{24 \pi^2 r^3} \, \varepsilon_{jkl}\,  {\rm Tr} \, \{\sigma_j \sigma_k \sigma_l \} \ =\ \frac 1{2\pi^2 r^3}\,.
  \ee
  Due to rotational symmetry, the value of the integrand is the same at all other points of the sphere. Multiplying \p{integrand} by $V_{S^3} = 2\pi^2 r^3$, we reproduce the result $q=1$.
  
  \vspace{1mm}
  Any 4-dimensional gauge field with the asymptotics \p{pure-krugl} has the topological charge
  $q=1$. But the configurations realizing the minimum of the action functional,
   \be
   \lb{act}
   S \ =\  \int_{R^4} {\rm Tr} \, \{ F \wedge *F\} \ =\ \frac 12 \int d^4x \,
    {\rm Tr} \, \{F_{\mu\nu} F_{\mu\nu} \}\,,
    \ee
    represent  a particular physical interest.  Such configurations with $q=1$ satisfy the self-duality condition, $ F = *F$ or
    \be
    F_{\mu\nu} \ =\ \frac 12 \varepsilon_{\mu\nu\rho\sigma}  F_{\rho\sigma}\,.
    \ee
    The standard BPST self-dual instanton with the center at the origin has the form
    \be
  \lb{BPST}
  A_\mu \ =\  \frac {\eta^a_{\mu\nu}  x_\nu\, \sigma^a}{r^2 + \rho^2} \,, \qquad F_{\mu\nu} \ =\ - \frac {2\rho^2 \eta^a_{\mu\nu} \sigma^a}{(r^2 + \rho^2)^2} \,,
   \ee
  where $\rho$ is a parameter having the meaning of the  instanton size.
  
  There are also self-dual or anti-self-dual configurations with an arbitrary integer topological charge $q$ \cite{ADHM}.   The values of $q$ give the degrees of  the corresponding maps. 
   
  \subsubsection{Exponential parameterization}
  Note (it will be helpful for us later) that one can also describe the mapping $S^3 \to SU(2)$ in an alternative way:
 \be
 \lb{g-ball}
 g\ = \ \exp\{-i\alpha_m \sigma_m\} =  \mathbb{1} \cos \alpha - \frac {i \sin \alpha}\alpha \alpha_m \sigma_m \,,
  \ee
  where  $0 \leq \alpha = \|\vec{\alpha}\| \leq \pi$.
The relation to the standard parameterization \p{g-krugl} is 
  \be
  \lb{varchange}
 \alpha_m \ =\ \frac {x_m}{\|\vec{x}\|} \arccos \left( \frac {x_4}r  \right) \,.
  \ee
  
 The 2-sphere $\|\vec{\alpha}\| = \pi$  maps onto one point of $S^3$ --- its south pole,    $g = - \mathbb{1}$.
 
 In these variables,

\be 
\lb{q3-ball}
  q \ =\ -\frac 1{4\pi^2} \int_{\alpha \leq \pi} d^3 \alpha \, {\rm Tr} \, \{g^{-1} \partial_1 g\, g^{-1} \partial_2 g\, g^{-1} \partial_3 g \}\,,
  \ee
  where
  \be
   g^{-1}\partial_m g = -iA_m  \ =\ -i \left[ \frac {\cos \alpha \sin \alpha}\alpha   \sigma_m -
  \frac {\sin^2 \alpha}{\alpha^2} \alpha_k \varepsilon_{kml} \sigma_l + 
  \frac {\alpha_m \alpha_k}{\alpha^2}\left( 1 - \frac {\cos \alpha \sin \alpha}\alpha \right) \sigma_k \right]\,.
  \ee
  The result of the calculation is, of course, the same:
  \be
  q \ =\ \frac 2\pi \int_0^\pi \sin^2 \alpha \, d\alpha \ =\ 1 \,. 
   \ee
 
 \subsection{$SO(4)$}
      Suppose now that we are interested in the gauge configurations belonging to $so(4)$ rather than $su(2)$ algebra. $so(4)$ is a direct sum of two $su(2)$ summands 
     with the generators $(\Sigma_1^A)_{ab} = -i\eta^A_{ab}/2$ and $(\Sigma_2^A)_{ab} = -i\bar \eta^A_{ab}/2$.\footnote{Then 
     \be
[\Sigma^A_1, \Sigma^B_1] = i \varepsilon^{ABC} \Sigma_1^C, \qquad  [\Sigma^A_2, \Sigma^B_2] = i \varepsilon^{ABC} \Sigma_2^C, \qquad [\Sigma_1^A, \Sigma_2^B] = 0\,,\nn {\rm Tr} \{\Sigma_1^A \Sigma_1^B\} = {\rm Tr} \{\Sigma_2^A \Sigma_2^B\} = \delta^{AB}\,, \qquad {\rm Tr} \{\Sigma_1^A \Sigma_2^B\} = 0 \,.\ee}
      An arbitrary $so(4)$ gauge field can be presented as
     $(A_\mu)_{ab} = B_\mu^A (\Sigma_1^A)_{ab} +  C_\mu^A  (\Sigma_2^A)_{ab}$.
     
      There are thus two types of instantons with the field densities
      \be
      \lb{2-inst}
     F_{\mu\nu}^{ab} \ =\  \frac {2i\rho^2  \eta^A_{\mu\nu} \eta^A_{ab}}{(r^2 + \rho^2)^2}   \quad {\rm and} \quad 
      F_{\mu\nu}^{ab} \ =\  \frac {2i\rho^2  \eta^A_{\mu\nu} \bar \eta^A_{ab}}{(r^2 + \rho^2)^2} \,.
       \ee
     In addition, there are two types of the antiinstantons:
      \be
     F_{\mu\nu}^{ab} \ =\  \frac {2i\rho^2 \bar \eta^A_{\mu\nu} \eta^A_{ab}}{(r^2 + \rho^2)^2}   \quad {\rm and} \quad 
      F_{\mu\nu}^{ab} \ =\  \frac {2i\rho^2 \bar \eta^A_{\mu\nu} \bar \eta^A_{ab}}{(r^2 + \rho^2)^2} \,.
       \ee

     Now note that, for the both types of instantons, the topological charge \p{Pont}, rewritten as
    \be
    q \ =\ -\frac 1{64\pi^2}\,  \varepsilon_{\mu\nu\alpha\beta} \int_{R^4} F^{ab}_{\mu\nu} F^{ab}_{\alpha\beta}  \, d^4x
     \ee   
     
       is equal to $q = 1$, while for the both types of antiinstantons it is $q = -1$. Consider now the second topological charge,
       
     \be
     \lb{Gauss4}
     \tilde {q} \ =\  - \frac 1 {16\pi^2} \varepsilon^{abcd}\int_{{\mathbb R}^4} F^{ab} \wedge F^{cd} \ =\    - \frac 1 {128\pi^2} \varepsilon^{abcd} \varepsilon_{\mu\nu\alpha\beta} \int  F_{\mu\nu}^{ab}   F_{\alpha\beta}^{cd} \,  d^4 x \,.
      \ee 
      Like the charge \p{Pont}, $\tilde q$ is a gauge invariant, as follows from the identity
      \be
      \lb{eps-ort}
      \varepsilon^{\tilde a \tilde b \tilde c \tilde d} O^{a \tilde a} O^{b \tilde b} O^{c \tilde c} O^{d \tilde d} \ =\  \varepsilon^{abcd}
       \ee
       for any matrix $O \in SO(4)$. 
      The  sign in \p{Gauss4} is chosen so that $\tilde q = 1$ for the first instanton in \p{2-inst}  (with $F \sim \eta\eta$). It is not difficult then to see that, for the second instanton, $\tilde q = -1$. Similarly, for the antiinstantons: $\tilde q$ have different signs for the different types.

\vspace{.5cm}
\bc
\begin{tabular}{l||c|c}
  configuration & $q$ & $\tilde q$ \\
\hline
\hline
First instanton & 1 & 1 \\
 \hline
First antiinstanton& -1 & -1 \\
\hline
Second instanton& 1 & -1 \\
\hline 
Second antiinstanton & -1 & 1 \\
\hline 

\end{tabular} 

\ec
\vspace{.4cm}

{\bf Table 2}. Topological charges of the $SO(4)$ instantons.

\vspace{3mm}

Note that the coefficient in the integral \p{Gauss4} for $\tilde q$ is the same as in the Gauss-Bonnet invariant, 
\be
 \lb{Bonnet}
 \chi^{D=4} \ =\ \frac1{128\pi^2} \, \varepsilon^{abcd} \varepsilon^{\mu\nu\alpha\beta} \int \sqrt{g}\, d^4x \, R^{ab}_{\mu\nu} \wedge R^{cd}_{\alpha\beta}  \,.
 \ee

 As for the Pontryagin index \p{Pont}, it has the same structure as the {\it Hirzebruch signature} invariant,
 \be
 \lb{Hirz}
 \tau \ =\ - \frac 1{24\pi^2} \int_M {\rm Tr} \{R \wedge R\}\,,
  \ee
  but the coefficients are different.

 \section{Conventional $SO(8)$ instantons}
 
 \setcounter{equation}0

 We  now go to eight dimensions and consider the gauge fields $A_{M= 1,\ldots,8}$ belonging to the algebra $so(8)$. The 8-dimensional counterpart of the topological charge \p{Pont} is the fourth Chern class \p{Chern4}.
  
  The construction is quite parallel to the standard instanton construction of the previous section. We note that the form Tr$\{F \wedge F \wedge F \wedge F\}$ is exact:
  \be 
  \lb{dom7}
 {\rm Tr} \{ F^4\} \ =\ d\omega_7, 
  \ee
  where\footnote{The expression \p{om7} follows from the general formula \cite{Zee} 
  \be
  \omega_{2n-1} \ =\ n \int_0^1 dt\, t^{n-1} {\rm Tr} \{A(dA - i t A^2)^{n-1} \}\,.
   \ee}
 \be
 \lb{om7}
 \omega_7 \ =\ {\rm Tr} \left\{AF^3 + \frac {2i}5 F^2 A^3 + \frac i5 AF A^2 F - \frac 15 FA^5 - \frac i{35} A^7     \right\}
  \ee
  (the products are understood here as wedge products).
  
  Let  $F$ decay rapidly at large distances so that the integral \p{Chern4} converges. Then $A$ represents at large distances a pure gauge \p{pure-gen}. The integral is reduced to \p{q7}.  
  
  The simplest  topologically nontrivial mapping $S^7  \to  SO(8) \simeq Semispin(8) $ has the form \p{g-Spin8}.  
   One of many explicit representations for the matrices $\Gamma_m$ satisfying \p{Cliff7} and \p{prod-gamma} is
            \be
            \lb{gamma-7}
            \Gamma_1 &=& -\sigma_2 \otimes \sigma_2 \otimes \sigma_2; \quad    \Gamma_2 = \mathbb{1} \otimes \sigma_1 \otimes \sigma_2; \quad   \Gamma_3 = \mathbb{1} \otimes \sigma_3 \otimes \sigma_2;
             \quad    \Gamma_4 = -\sigma_1 \otimes \sigma_2 \otimes \mathbb{1}; \nn
               \Gamma_5 &=& \sigma_3 \otimes \sigma_2 \otimes \mathbb{1}; \quad
                 \Gamma_6 = \sigma_2 \otimes \mathbb{1} \otimes \sigma_1; \quad
                   \Gamma_7 = \sigma_2 \otimes \mathbb{1} \otimes \sigma_3 \,.
                    \ee
    Seven matrices $\Gamma_m$ are among the generators of the group $Semispin(8)$. 21 other generators are 
      \be
      \lb{gen7}
    \Sigma_{mn} \ =\ \frac  i2 [\Gamma_m, \Gamma_n] \,.
     \ee
     One can also define 
     \be
         \lb{gamma-M}
         \Gamma_M \ =\ (\vec{\Gamma}, i), \qquad  \Gamma_M^\dagger  \ =\ (\vec{\Gamma}, -i)\,,
          \ee 
          Then the whole set of the generators of $Semispin(8)$ can be presented as
        \be
        \lb{gen8}
         \Sigma_{MN} \ =\ \frac i2  (\Gamma_M \Gamma_N^\dagger - \Gamma_N \Gamma_M^\dagger)\,.
          \ee   
     
    The embedding \p{g-Spin8} is quite analogous to \p{g-ball}, but the difference is that in this case    not all the generators of
    $Semispin(8)$ are   engaged and it is an injective rather than surjective map. The sign in \p{g-Spin8} is chosen such that the topological charge \p{q7} is positive.
    
     The matrix \p{g-Spin8} treated as an element of $SO(8)$ describes the rotation by the angle $\alpha$ in the plane spanned by the unit vectors $\vec{e}_8$ and $\alpha_j \vec{e}_j/\alpha$. This brings the axis $Ox_8$ in a new position. An arbitrary $SO(8)$ rotation can be presented as the superposition of such rotation and an $SO(7)$ rotation ``around" the new axis. Similarly, an arbitrary element of $Semispin(8)$ can be presented as
  \be
  \lb{repr-Spin8}
  g \in Semispin(8) \ =\ \exp\{i \alpha_j \Gamma_j\} \times \left[ g \in Spin(7) \right]\,,
   \ee
   $0 \leq \|\alpha_j\| \leq \pi$.
    
   Note that, by the variable change 
    \be
  \lb{varchange8}
 \alpha_m \ =\ \frac {x_m}{\|\vec{x}\|} \arccos \left( \frac {x_8}r  \right) \,,
  \ee
     \p{g-Spin8} can be presented in the form analogous to \p{g-krugl}:
    
 \be
 \lb{g8-krugl}
 g_8 \ =\ \frac {i \Gamma_M^\dagger x_M}{r} \ =\  \frac {\mathbb{1} x_8  + i x_m \Gamma_m}{r}
 \ee
 with $r^2  =  (x_M)^2$.   The unit vector $n_M = x_M/r$ gives a natural ``round" parameterization of $S^7$. 
 
 \vspace{1mm}
 
  The embedding \p{g-Spin8}, \p{g8-krugl} represents a   section of the fiber bundle
    \be
    \lb{fiber8}
   Semispin(8)  \ \stackrel{Spin(7)}{\longrightarrow} \ S^7 \,.
    \ee
    Note that, for $Semispin(8)$,  the representation \p{repr-Spin8} is not unique. For example, a unit element of $Semispin(8)$ can be represented in two ways:
    \be
    \lb{2ways}
    \mathbb{1}_{Semispin(8)} \ =\  \mathbb{1}_{S^7} \times  \mathbb{1}_{Spin(7)} \ =\  (-\mathbb{1})_{S^7} \times  (-\mathbb{1})_{Spin(7)}\,.
     \ee
     However, an element of $Spin(8)$, the double covering of $Semispin(8)$, is represented by {\it two} distinct matrices,
     $ \exp\{i \alpha_j \Gamma_j\}$   and $g_7 \in Spin(7)$, 
     rather than by their product, and {\it this} representation is unique. 
     
     As a result, the whole group $Spin(8)$ represents the direct product $S^7 \times Spin(7)$, the bundle \p{fiber8} is trivial and the embedding 
      \p{g-Spin8} is its global section.\footnote{The triviality of the bundle \p{fiber8} is a well known fact. It follows also from a construction based on a representation of $SO(8)$ elements as certain octonion multiplications \cite{Steen}. It is this fact that allows one to derive the result \p{ZZ}.} 
     
     \vspace{1mm}
 
 Substituting \p{g8-krugl} into \p{pure-gen}, we derive the expression for the asymptotic gauge field:  
 \be
 \lb{pure-krugl-8}
 A_M \ =\ \frac { \Sigma_{NM} x_N}{r^2} 
  \ee
 
 Now, pick up the north pole of $S^7$, $x_M = (\vec{0}, r)$. At this point, $A_8 = 0$ and
 \be
 g^{-1} \partial_m g  = -iA_m \ =\ \frac {i\Gamma_m}r\,.
  \ee
  The integrand in \p{q7} is
  \be
  \lb{integrand8}
   -(i)^7 \frac {7!}{r^7}     {\rm Tr} \, \{\Gamma_1 \Gamma_2 \Gamma_3 \Gamma_4\Gamma_5\Gamma_6\Gamma_7 \} \ =\ 
  - \frac   {8!}{r^7} 
    \ee
    [the first minus sign in \p{integrand8} comes from $\epsilon_{81234567} = -1$, cf. Eq. \p{q-krugl}].
    
     Multiplying it by the volume of $S^7$, 
    \be
    \lb{volS7}
 V_{S^7} \ =\ \frac {\pi^4 r^7}3\,,
  \ee
  and substituting in \p{q7}, we derive
  \be
  q \ =\ 1\,. 
  \ee

      There is one particular field $A_M(x \in \mathbb{R}^8)$ in the sector $q=1$ where the field density satisfies the {\it nonlinear} self-duality condition
      \be
       \lb{8-self-dual}
       F \wedge F \ =\  \star (F\wedge F)
        \ee
      It reads
      \be
 \lb{inst8}
 A_M(x) \ =\ \frac { \Sigma_{NM} x_N}{r^2+ \rho^2} \,,
  \ee  
  going to \p{pure-krugl-8} in the large $r$ limit.
  
  The field density is
  \be
  \lb{inst8-F}
  F_{MN} \ =\ \partial_M A_N - \partial_N A_M - i[A_M, A_N] \ =\ \frac {2\rho^2 \Sigma_{MN}}{(r^2 + \rho^2)^2}\,.
   \ee
  
  The gauge field configuration \p{inst8}, \p{inst8-F} is what is called  $SO(8)$ instanton  \cite{Keph,Tchr}. It realizes the minimum of the functional
  \be
  \lb{act8} 
  \int_{R^8} {\rm Tr} \, \{F \wedge F \wedge \star( F \wedge F) \} 
   \ee 
   in the sector $q=1$.
  The analogy with the BPST instanton \p{BPST} is straightforward.

  Starting from the map
  \be
  \lb{g-Spin8-anti}
  g \ =\ \frac{\mathbb{1} x_8 - i  x_m \Gamma_m}r
   \ee
   with the topological charge $q = -1$, we arrive at the {\it antiinstanton}
    \be
 \lb{antiinst8}
 A_M(x) \ =\ \frac {\tilde \Sigma_{NM} x_N}{r^2+ \rho^2} \,,
  \ee  
  where $\tilde \Sigma_{NM} = \frac i2 (\Gamma_M^\dagger\Gamma_N  - \Gamma_N^\dagger \Gamma_M )$.
  The corresponding field density is anti-self-dual, $ F \wedge F \ =\  -\star (F\wedge F)$.
   
 The configurations with an arbitrary topological charge are also known \cite{ADHM-8}. First of all, one can write a natural 8-dimensional generalisation of the 't Hooft Ansatz:
  \be
 \lb{Hooft-Ans}
 A_M(x) \ &=&\ - \tilde \Sigma_{NM} \partial_N \ln \left( 1 + \sum_{i=1}^q \frac {\rho_i^2}{(x_P - a_P^{(i)})^2} \right) \qquad {\rm and} \nn
 \quad A_M(x) \ &=&\ - \Sigma_{NM} \partial_N \ln \left( 1 + \sum_{i=1}^q \frac {\rho_i^2}{(x_P - a_P^{(i)})^2} \right)\,.
  \ee  
The topological charges of these configurations are $q$ and $-q$, respectively. When $q=1$, we reproduce the instanton of \cite{Keph,Tchr}, but written in the singular gauge. One can also write more general multiinstanton solutions generalizing the AHDM construction.

   Note, however, that all these configurations are characterized by only one integer \p{Chern4}, whereas the homotopy group \p{ZZ} implies the existence of the topological classes characterized by {\it two} integers. Such {\it generalized} $SO(8)$ instantons have not been studied so far. They are related to $Spin(7)$ instantons to be discussed in the next section.

  \section{$Spin(7)$ instantons}
 \setcounter{equation}0
 
 The generators of $Spin(7) \subset Semispin(8)$ are given in \p{gen7}. As was noted above, 
 the embedding \p{g-Spin8}, which involves seven generators of $Semispin(8)$ that are  not in the set \p{gen7}, is a global section in the fiber bundle \p{fiber8}.

    It is known, however, that the sphere $S^7$ also represents  a base for another fiber bundle 
     \be
    \lb{fiber7}
    Spin(7)  \ \stackrel{G_2}{\longrightarrow} \ S^7 \, .
     \ee 
     Each fiber of this bundle represents a subgroup $G_2 \subset Spin(7)$ that does not transform a particular 8-component real spinor (it is one of the definitions of $G_2$). The set of all such spinors of unit norm gives the base $S^7$ of the bundle.
     
     The bundle \p{fiber7} implies a nontrivial homotopy 
   \be
   \lb{pi7-Spin7}
   \pi_7[Spin(7)] \ =\ \mathbb{Z}\,.
    \ee
   Indeed, one can write the following exact sequence:
   \be
   \pi_7[G_2] \ \stackrel{1}{\longrightarrow} \ \pi_7[Spin(7)] \ \stackrel{2}{\longrightarrow} \ \pi_7[S^7] 
   \stackrel{3}{\longrightarrow} \pi_6[G_2]
    \ee  
    or
     \be
   0 \ \stackrel{1}{\longrightarrow} \ \pi_7[Spin(7)] \ \stackrel{2}{\longrightarrow} \ \mathbb{Z} 
   \stackrel{3}{\longrightarrow} Z_3
    \ee  
    (the homotopies of $G_2$ can be found in \cite{Mimura}).
    
    The image of the map{\bf 2} coincides with the kernel of the map {\bf 3}, which is $\mathbb{Z}$. Hence,
    $\pi_7[Spin(7)]$ is $\mathbb{Z}$ or larger. But the kernel of the map {\bf 2} coincides with the image of the map {\bf 1}, which is trivial. Hence $ \pi_7[Spin(7)]$  cannot be larger than  $\mathbb{Z}$, and we are led to
    \p{pi7-Spin7}. 
    
    \vspace{1mm}
    
    In the previous section, we discussed the embedding \p{g-Spin8} of $S^7$ into $Semispin(8)$, which was topologically nontrivial. But only two elements of
    \p{g-Spin8}: $g = \mathbb{1}$ and $g = -\mathbb{1}$ were also the elements of $Spin(7)$. Probably, the main new result of this paper is an explicit construction for a topologically nontrivial embedding $S^7 \to Spin(7)$. 
    
     \begin{figure} [ht!]
      \bc
    \includegraphics[width=.6\textwidth]{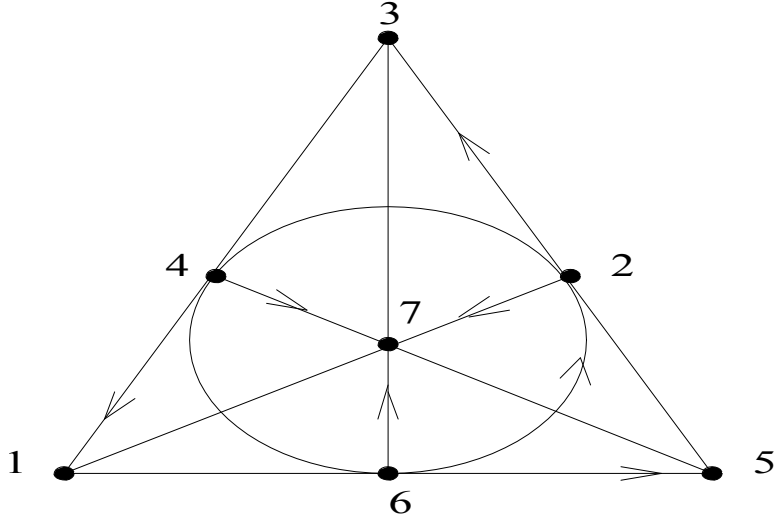}                  
     \ec
    \caption{Fano graph.}        
 \label{Fano}
    \end{figure}

  Consider the expression  
  \be
  \lb{g-Spin7hr}
  g \ =\ \exp \left\{- \frac 12 \alpha_m f_{mnk} \Gamma_n \Gamma_k \right\}\,,
   \ee
   where $f_{mnk}$ are the structure constants of the octonion algebra,
   \be
   \lb{fmnk}
   f_{165} \ =\  f_{134} \ =\  f_{127} \ =\  f_{235} \ =\  f_{246} \ =\  f_{367} \ =\  f_{475}\ =\ 1
    \ee
    and all other nonzero elements of $f_{mnk}$ are restored by antisymmetry (see the Fano graph on Fig. \ref{Fano}).
    \p{g-Spin7hr} is an element of  $Spin(7)$. Its logarithm represents a linear combination of seven $Spin(7)$ generators, $T_m \ =\   \frac 12 f_{mnk} \Sigma_{nk}$.
   Explicitly:
   \be
   \lb{Tj}
   &&T_1 \ =\ \Sigma_{34}  + \Sigma_{27}  - \Sigma_{56}, \qquad T_2 \ =\  \Sigma_{35} + \Sigma_{46}  -\Sigma_{17},  \qquad   T_3 \ =\   \Sigma_{67} -\Sigma_{14} - \Sigma_{25}, \nn
   &&T_4 \ =\ \Sigma_{13} - \Sigma_{26} - \Sigma_{57}, \qquad T_5 \ =\ \Sigma_{23} + \Sigma_{16} + \Sigma_{47}, \nn
   &&T_6 \ =\  \Sigma_{24}   -\Sigma_{15}- \Sigma_{37}, \qquad T_7 \ =\ \Sigma_{12} + \Sigma_{36} - \Sigma_{45}.
    \ee 
    Each $T_m$ is a linear combination of three {\it commuting}  generators of $Spin(7)$.

  \begin{thm}
  The expression \p{g-Spin7hr} considered in the range $0 \leq \|\vec{\alpha}\| \leq \pi$
  represents an embedding of $S^7$ into $Spin(7)$. Locally, it represents a section of the fiber bundle \p{fiber7}.
   \end{thm}
   
   \begin{proof}

   To prove that \p{g-Spin7hr} is an embedding of $S^7$, we only have  to prove that, for any $\vec{\alpha}$ of norm $\pi$, $g(\vec{\alpha}) = -\mathbb{1}$, as it was for the mappings \p{g-ball} and \p{g-Spin8}. To this end, we notice that $f_{mnk}$ is an invariant tensor of $G_2$ and hence \p{g-Spin7hr} is a $G_2$ invariant depending on a 7-dimensional vector $\alpha_m$. $G_2$ may be defined as the subgroup of $SO(7)$  involving only the rotations leaving invariant the ``mixed product" of three different 7-vectors, $ I \ =\   f_{mnk} A_m B_n C_k$. But we have only one vector at our disposal, and  \p{g-Spin7hr} may depend only on the norm of $\alpha_m$, 
   but not on its orientation. Choosing $\alpha_m = (\pi, 0,0,0,0,0,0)$, we derive
   \be
   g(\|\alpha_m\| = \pi) = e^{i\pi T_1}  = \exp\{-\pi\Gamma_3 \Gamma_4\} \exp\{-\pi\Gamma_2 \Gamma_7\}\exp\{\pi\Gamma_5 \Gamma_6\} \ 
    = \ (-\mathbb{1})^3 = -\mathbb{1}.
    \ee

       To prove that \p{g-Spin7hr} represents a local section of the bundle \p{fiber7}, we need to prove that the 
   generators \p{Tj}  are {\it not} in the subalgebra $g_2 \subset spin(7)$, but are orthogonal to it:
    \be
    \lb{orto}
     {\rm Tr}\{ T_j h\} \ =\ 0, \qquad {\rm if} \ h \in g_2 \,.
     \ee
    Choosing  $f_{mnk}$ as in \p{fmnk}, the $G_2$ invariant acquires the following explicit form  
      
       \be
       \lb{fABC}
    I \ =\   f_{jkl} A_j B_k C_l \ =\ [165] + [134] + [127] + [235] + [246] + [367] + [475],
     \ee
     where
     \be 
     \lb{165}
     [165] \ =\ (A_1 B_6 - A_6 B_1)C_5 + (B_1 C_6 - B_6 C_1)A_5 + (C_1 A_6 - C_6 A_1)B_5\,,
       \ee
       etc.     The explicit form of the generators of $G_2 \subset SO(7)$, which do not transform $I$, is 
       \be
       \lb{gen-G2}
       \tilde \Sigma_{34} +  \tilde \Sigma_{56}, \quad \tilde \Sigma_{35} -  \tilde \Sigma_{46}, \quad    \tilde\Sigma_{25} - \tilde\Sigma_{14}, \quad \tilde\Sigma_{13} +  \tilde\Sigma_{26}, \ 
       \tilde\Sigma_{23} -  \tilde\Sigma_{16},  \quad \tilde\Sigma_{15} +  \tilde\Sigma_{24}, \quad \tilde\Sigma_{12} -  \tilde\Sigma_{36}, \nn
     \tilde\Sigma_{34} -  \tilde\Sigma_{27}, \quad \tilde\Sigma_{35} +  \tilde\Sigma_{17}, \quad \tilde\Sigma_{14} +  \tilde\Sigma_{67}, \quad \tilde\Sigma_{13} +  \tilde\Sigma_{57}, \quad 
       \tilde\Sigma_{23} -  \tilde\Sigma_{47}, \quad  \tilde\Sigma_{37}  - \tilde\Sigma_{15} , \quad \tilde\Sigma_{12} +  \tilde\Sigma_{45}, \ \ \ \ 
      \ee
      where $\tilde \Sigma_{mn}$ are the counterparts of \p{gen7} in the vector representation.
      
      Consider e.g. the action of the operator $\tilde \Sigma_{34}$ on $I$. The structures [165], [134] and [127] are left invariant. But four other structures are not:
       $$ 
       \tilde \Sigma_{34}\, I \ \propto \ [245] - [236] + [467] - [375].
        $$
        One can observe, however, that this  exactly coincides with $-\tilde \Sigma_{56} \, I$ and  $\tilde \Sigma_{27} \, I$, so that $I$ does not transform under the action of the generators $ \tilde \Sigma_{34} +  \tilde \Sigma_{56}$ or $  \tilde\Sigma_{34} -  \tilde\Sigma_{27}$. On the other hand,
        $$ ( \tilde \Sigma_{34} +  \tilde \Sigma_{27} -  \tilde \Sigma_{56}  ) \, I \ \propto \ 3([245] - [236] + [467] - [375]) \ \neq \ 0. $$
        The same is true for six other generators in \p{Tj} and their arbitrary linear combinations.
        
        The property \p{orto} can be verified explicitly.
        
        Hence the generators \p{gen-G2} together with \p{Tj} form an orthogonal basis in $spin(7)$ and  
        an arbitrary element of $Spin(7)$ can be represented as
     \be
  \lb{repr-Spin7}
  g \in Spin(7) \ =\ g_7 \times \left[ g \in G_2 \right]\,,
   \ee
     where $g_7$ is the embedding \p{g-Spin7hr}. 
     
     We have thus proven that \p{g-Spin7hr} represents a local section of the bundle \p{fiber7}. We will see a bit later that the representation \p{repr-Spin7} is not unique and hence    \p{g-Spin7hr} is not its global section.
    
     \end{proof}

   By applying an approriate $G_2$ rotation, the exponential in \p{g-Spin7} can be disentangled for an arbitrary $\vec{\alpha}$, the condition $\alpha \equiv \|\vec{\alpha}\| = \pi$ is not necessary. For $\alpha_m = \ (\alpha_1,0,0,0,0,0,0)$, we derive
   \be
   \lb{g-al1}
   g_7 \ =\ (\cos \alpha_1 \, \mathbb{1} - \sin \alpha_1 \, \Gamma_3 \Gamma_4) (\cos \alpha_1  \, \mathbb{1} - \sin \alpha_1 \, \Gamma_2 \Gamma_7) (\cos \alpha_1  \, \mathbb{1} + \sin \alpha_1 \, \Gamma_5 \Gamma_6)   \nn
   = \ \cos^3 \alpha_1\, \mathbb{1} + i \cos^2 \alpha_1 \sin \alpha_1 \, T_1 +  \cos \alpha_1 \sin^2 \alpha_1 \, \Gamma_1 T_1  + i \sin^3 \alpha_1 \, \Gamma_1  .
    \ee
    The third and the fourth term in the R.H.S. of this identity were derived using \p{prod-gamma}.
    
    For an arbitrary $\vec{\alpha}$, we obtain
    \be
    \lb{g-explicit}
    g_7  = \ \cos^3 \alpha\, \mathbb{1} + i \cos^2 \alpha \sin \alpha \, \frac {  \alpha_m  T_m}\alpha  +  \cos \alpha \sin^2 \alpha \,  \frac {\alpha_m \alpha_n \Gamma_m  T_n}{\alpha^2}  + i \sin^3 \alpha \, \frac {\alpha_m\Gamma_m} \alpha .
     \ee
  This form is convenient for rewriting $g$ in the Cartesian coordinates $x_M$. Using \p{varchange8}, we derive
   \be 
   \lb{g-Spin7-round}
   g_7 \ =\ \frac { x_8^3\, \mathbb{1} + i x_8^2 \, x_m T_m  + x_8 \, x_m   x_n  \Gamma_m \, T_n  + i x_m x_m\,  x_n \Gamma_n  } {r^3}.
    \ee 
    
    We can explain now why \p{g-Spin7} is not a global section.
    
    The subgroup $G_2 \subset Spin(7)$ spanned by the generators $ \Sigma_{34} +  \Sigma_{56}$ etc. does not transform the spinor
    \be
    \psi_0 \ =\ \left(  \begin{array}{c} 0 \\ 1 \end{array} \right) \otimes  \left(  \begin{array}{c} 0 \\ 1 \end{array} \right) \otimes 
     \left(  \begin{array}{c} 0 \\ 1 \end{array} \right)\,.
      \ee
      Let us direct $\alpha_m$ along the first axis and consider the action of a group element \p{g-al1} on this spinor. One can derive
      \be
\lb{triple}
      g(\alpha_1) \psi_0 \ =\ \cos 3 \alpha_1\, \psi_0 + \sin 3 \alpha_1 \,\psi_1\,,
       \ee
       where 
       \be
       \psi_1 \ =\ \left(  \begin{array}{c} 1 \\ 0 \end{array} \right) \otimes  \left(  \begin{array}{c} 1 \\ 0 \end{array} \right) \otimes 
     \left(  \begin{array}{c} 1 \\ 0 \end{array} \right)\,. 
      \ee
      We see that  three different elements in the set \p{g-al1} with $\alpha_1 = 0, \pm 2\pi/3$,  and many more different elements in the set \p{g-explicit} keep $\psi_0$ intact\footnote{Any element belonging to $S^6$ with $\|\alpha_m\| = 2\pi/3$ does so.} and hence belong to one and the same $G_2$ fiber. For a global section, this would not be possible. 
       In contrast to the bundle \p{fiber8} that could be trivialized, the bundle \p{fiber7} cannot. $Spin(7)$ is not homeomorphic to the direct product $S^7 \times G_2$ (indeed, homotopy groups of these two spaces are different \cite{Mimura}) and a global section of the bundle \p{fiber7} does not exist.

      \vspace{1mm}

Our next task is to evaluate the topological charge \p{q7} 
for the map \p{g-explicit}. A nonzero value of the charge will assure the topological nontriviality of the mapping.\footnote{Not any mapping of a sphere into a Lie group is topologically nontrivial. For example, the expression
\be
\lb{g-S2}
g \ =\ \exp\{i(\alpha_1 \sigma_1 + \alpha_2 \sigma_2)\} 
 \ee
 with $0 \leq \sqrt{\alpha_1^2 + \alpha_2^2} \leq \pi$ realizes a mapping $S^2  \to SU(2)$ related to the fiber bundle
 $SU(2) \ \stackrel{U(1)}{\longrightarrow} \ S^2$. However, there is no topological charge associated with the mapping \p{g-S2} and this mapping is topologically trivial (contractible to a point by a continuous deformation): $\pi_2[SU(2)] = 0$.}

The calculation of $q$ for the map \p{g-Spin7} is somewhat more complicated than for the map \p{g-Spin8}, but we can capitalize again on the rotational invariance of the integrand, calculate it for 
$\alpha_m = \alpha \delta_{m1}$, multiply the integrand 
    by the volume of $S^6$ and then integrate it over $\alpha$. With the help of Mathematica, we derived the result
 \be
  \lb{q-intal7}
  q[g_7] \ =\  \frac {512}{75\pi}\, \int_0^\pi \sin^6 \alpha \cos^2 \alpha (2 \cos^4 \alpha + 5 \cos^2 \alpha + 2) \, d\alpha \ =\ 1\,.
   \ee
For the embedding $g_8$, it was not necessary, but we could also calculate the topological charge $q[g_8]$ by the same method  to derive
 \be
  \lb{q-intal8}
  q[g_8] \ =\  \frac {16}{5\pi}\, \int_0^\pi \sin^6 \alpha \, d\alpha \ =\ 1\,.
   \ee

For the maps $g = g_7^2$ and $g = g_8^2$, the topological charge is $q =2$. One can be convinced that 
$q[g_7^2]$ and $q[g_8^2]$ are given by the integrals where the integrands are related to the integrands in \p{q-intal7}, \p{q-intal8} as 
 \be
f_2(\alpha) \ =\ 2 f_1 (2\alpha) \,.
 \ee
 
 Having the expression \p{g-explicit} for $g_7(x)$ in hand, one can derive the expressions for the gauge potentials
   \be
\lb{asympt}
 A_M = ig_7^{-1} \partial_M g_7
  \ee
  at the distant 7-sphere. However, they 
  are rather complicated and we will not quote them here.

All topologically nontrivial $spin(7)$ gauge field configurations belonging to the sector $q=1$ have the asymptotics \p{asympt}. There is a distinguished configuration realizing the minimum of the functional \p{act8} and satisfying the nonlinear self-duality conditions \p{8-self-dual}. This is the $Spin(7)$ instanton. Its explicit form is, however, difficult to find.

The embedding
 \be
  g \ =\ \exp \left\{ \frac 12 \alpha_m f_{mnk} \Gamma_n \Gamma_k \right\}\,,
   \ee
   has the topological charge $q  = -1$. It brings about the anti-self-dual antiinstanton.

Note that the maps $S^7 \to Spin(7)$  are also maps  of $S^7$  into $SO(8)$.
As was mentioned above, the generalized topologically nontrivial maps  $S^7 \to SO(8)$ and the corresponding instantons are characterized by two integer charges $(q,\tilde q)$. The Chern classes $q$ of the instanton \p{inst8}  and of the instanton, corresponding to the embedding \p{g-Spin7}  are both equal to one.

But the values of the second charge \p{Gauss-S7} are different for these two embeddings.

Let us first explain how \p{Gauss-S7} is derived. The integrand in \p{Gauss8} is an 8-form, which is closed in $\mathbb{R}^8$. Due to the trivial topology of $\mathbb{R}^8$, it is also exact,
 \be
\varepsilon^{a_1  \cdots  a_8} \,F^{a_1 a_2} \wedge F^{a_3 a_4} \wedge F^{a_5 a_6}  \wedge F^{a_7 \,a_8} \ =\ d \tilde \omega_7\,,
\ee
and the integral \p{Gauss8} only depends on the gauge field $A_{S^7}$ at the distant 7-sphere. 

We can find $\tilde \omega_7$ by the same method as $\omega_7$ in \p{om7} was found.
Consider instead of $\mathbb{R}^8$ an open 8-dimensional ball of unit radius and consider there the gauge potential $A_{B^8} = r A_{S^7}$, where $r$ is the distance from the center.
The field density form is then
\be
F_{B^8} \ =\ dr \wedge A_{S^7} + r dA_{S^7} - ir^2 A_{S^7} \wedge  A_{S^7} \,.
\ee
We derive
\be
\varepsilon^{a_1  \cdots  a_8} \,F^{a_1 a_2} \wedge F^{a_3 a_4} \wedge F^{a_5 a_6}  \wedge F^{a_7 \,a_8}  =  4 \varepsilon^{a_1  \cdots  a_8} dr \wedge A_{S^7}^{a_1a_2}   \wedge (r dA_{S^7} - ir^2 A_{S^7} \wedge  A_{S^7} )^{a_3a_4}  \nn
\wedge
(r dA_{S^7} - ir^2 A_{S^7} \wedge  A_{S^7} )^{a_5a_6} \wedge (r dA_{S^7} - ir^2 A_{S^7} \wedge  A_{S^7} )^{a_7a_8}  \,.
 \ee
We are interested only with $\tilde \omega_7$ at the boundary of the ball and assume that the field density vanishes there. Using this, we can substitute $  iA_{S^7} \wedge  A_{S^7}$
for $d A_{S^7}$ and integrate over $r$. We derive
 \be
\tilde \omega_7[{\rm boundary}] \ =\ 
- 4i \int_0^1 (r^2 - r)^3 \,  dr     \ \varepsilon^{a_1  \cdots  a_8} A^{a_1 a_2} \wedge \nn
 (A \wedge A)^{a_3a_4} \wedge (A \wedge A)^{a_5a_6}  \wedge (A \wedge A)^{a_7a_8} \ = \nn
 \frac i{35}  \, \varepsilon^{a_1  \cdots  a_8} A^{a_1 a_2} \wedge (A \wedge A)^{a_3a_4} \wedge (A \wedge A)^{a_5a_6}  \wedge (A \wedge A)^{a_7a_8} \,.
\ee
This gives \p{Gauss-S7}.

The calculation gives $\tilde q[g_8] = 1$, while

 \be
  \lb{tild-q-intal7}
 \tilde q[g_7] \ =\  -\frac {512}{25\pi}\, \int_0^\pi \sin^6 \alpha \cos^2 \alpha (2 \cos^4 \alpha + 5 \cos^2 \alpha + 2) \, d\alpha \ =\ -3\,.
   \ee
We see  the same integral as \p{q-intal7} but with an additional factor -3 !
Obviosly, the presence of this factor is related to the fact that \p{g-Spin7} is a local rather than global section of the bundle \p{fiber7} and the presence of triple intersections of the embedding \p{g-Spin7} with any $G_2$ fiber, as displayed in Eq.\p{triple}.

Consider now the composite map $g = g_7 g_8$ and calculate its topological charges. We derive

\be
  \lb{q-intal78}
  q[g_7 g_8] \ =\  - \tilde q[g_7 g_8]  \ =\ \frac {32}{5\pi}\, \int_0^\pi \sin^6 (2\alpha) \, d\alpha \ =\ 2\,,
   \ee
 the same for $g = g_8 g_7$  and the same up to the sign for $g = g_8^{-1} g_7^{-1}$.  

 Finally, for the map $g = g_7^{-1} g_8$ or $g = g_8 g_7^{-1}$  the topologocal charges are
 \be
q[g_7 g_8^{-1}] \ =\ \int_0^\pi 0\, d\alpha \ =\ 0
 \ee
and 
\be
\tilde q[g_7^{-1} g_8]  \ =\  \frac {64}{5\pi}\, \int_0^\pi \sin^6 (2\alpha) \, d\alpha \ =\ 4 \,.
 \ee
Changing $g \to g^{-1}$ reverses the sign of the both charges.

These results were put together in Table 1. 

We see that the topological charges of the composite maps coincide with the sums of the individual topological charges, as it, of course, should be. For a generic composition, $g = g_7^n g_8^m$,
\be
q(g_8^n g_7^m) \ =\ n+m, \qquad \tilde q(g_8^n g_7^m) \ =\ n - 3m\,.
\ee

  \section{Other groups and dimensions}
 \setcounter{equation}0

As was mentioned, $\pi_{4n-1}[SO(4n)] = \mathbb{Z} \times \mathbb{Z}$, which inplies that the  mappings $S^{4n-1} \to SO(4n)$ and associated instantons are characterized by two integer invariants. It is not too difficult to generalize for higher $n$ the conventional $SO(8)$ instantons \cite{Tchr,OSe}.  All the formulas are basically the same as in Sect. 3.1. where one has to replace 7-dimensional gamma matrices by $4n-1$-dimensional gamma matrices $\Gamma_{m = 1,\ldots,4n-1}$. Note that these matrices have large dimension $2^{2n-1} \times
2^{2n-1}$ and the formulas like \p{inst8}, \p{inst8-F} describe now gauge potentials and field densities in the spinor representation, which is not equivalent for $n \neq 2$ to the vector representation.

The (anti-)instantons realize here the minimum of the functional
\be
  \lb{act-} 
  \int_{R^{4n}} {\rm Tr} \, \{\overbrace{F \wedge \ldots \wedge F}^n\ \star( \overbrace{F  \wedge \ldots \wedge F}^n ) \}
   \ee   
The field densities satisfy the following nonlinear self-duality conditions
\be
       \lb{4r-self-dual}
       \overbrace{F \wedge \ldots \wedge F}^n \ =\  \pm \star ( \overbrace{F \wedge \ldots \wedge F}^n)\,.
\ee

There should exist also another type of $Spin(4n)$ instantons which are simultaneosly $Spin(4n-1)$ instantons and which are associated with the multidimensional analogs of the mapping \p{g-Spin7}. However, explicit expressions for such mappings, not speaking of the explicit expressions for the instantons are not known. 

We can stay in $4n$ dimensions but consider  the gauge group that is smaller than $SO(4n)$. For example, one can take   the group $SU(2n)$. In this case, $\pi_{4n-1}[SU(2n)] = \mathbb{Z}$ and there is only one type of instanton  
\cite{Takesue}.

Alternatively, one can take the still smaller gauge group $Sp(n) \subset SU(2n)$, which also enjoys a nontrivial homotopy $\pi_{4n-1}[Sp(n)] = \mathbb{Z}$ related to the fiber bundle
$$ Sp(n)  \ \stackrel{Sp(n-1)}{\longrightarrow} \ S^{4n-1}\,,$$
which implies the presence of instantons.\footnote{One can notice in this regard that $Sp(n)$ is the holonomy group of a hyper-K\"ahler manifold of dimension $4n$. This suggests that the instantons with the gauge group $Sp(n)$ exist not only in $R^{4n}$, but in any $4n$-dimensional hyper-K\"ahler manifold.}

It would be interesting to make the ADHM-inspired construction of Ref. \cite{Takesue} for $SU(2n)$ instantons more explicit and generalize it to the group $Sp(n)$.

 One can also construct topologically nontrivial gauge field configurations in  $R^{4n-2}$. In  6 dimensions, they are associated with the nontrivial $\pi_5[SU(3)] = \mathbb{Z}$ and the fiber bundle $SU(3) \stackrel{SU(2)}{\longrightarrow} S^5$.
 In this case, each fiber represents the  subgroup $SU(2) \subset SU(3)$ that leaves intact a unit complex vector,
    \be
 \lb{V-S5}
  V \ =\ \left( \begin{array}{c} \alpha_1 \\ \alpha_2 \\ \alpha_3 \end{array} \right) , \qquad |\alpha_1|^2 + |\alpha_2|^2 + |\alpha_3|^2 = 1,
    \ee
 and the base $S^5$ is the set of all such vectors.  An explicit expression for the mapping $S^5 \longrightarrow SU(3)$ is not so nice as in \p{g-Spin7}, but it is known \cite{Khanna}. If $|\alpha_1| > 0$, the vector \p{V-S5}  maps into the matrix\footnote{There are two similar expressions $g_{2,3}$ for the regions  $|\alpha_2| > 0$ and  $|\alpha_3| > 0$. $S^5$ is described as  a union of these  three regions.}
 
 \be
 \lb{g-SU3}
 g_1 \ =\ \left(  \begin{array}{ccc}  -\frac{ \alpha_1 \alpha_2^*} {\sqrt{1 - |\alpha_2|^2}} & - \frac{ \alpha_3^*} {\sqrt{1 - |\alpha_2|^2}}  & \alpha_1 \\
  \sqrt{1 - |\alpha_2|^2} & 0 & \alpha_2 \\
 -\frac{\alpha_3 \alpha_2^*} {\sqrt{1 - |\alpha_2|^2}}  & \frac{\alpha_1^*} {\sqrt{1 - |\alpha_2|^2}}   & \alpha_3 
   \end{array} \right) \ \in \ SU(3)\,.
   \ee
  With this expression in hand, one can in principle calculate $A = i g^{-1} dg$ and determine the topological charge of the mapping, which should give
  
   \be
   \lb{qS5}
   q \ =\ \pm \frac 1{480 \pi^3}  \int_{S^5} {\rm Tr} \{ (g^{-1} d g)^5 \} \ = \ \pm 1 \,.
    \ee

    The coefficient in \p{qS5} is obtained from  the expression \p{Chern-D}   for the third Chern class,
     \be
     \lb{Chern3}
      q \ =\ \pm \frac 1{48 \pi^3} \int  {\rm Tr} \{F \wedge F \wedge F \}\,,
       \ee
      and the identity
       \be
       {\rm Tr} \{F^3\} \ =\ d \omega_5
        \ee
        with 
        \be 
        \omega_5 \ = \ {\rm Tr}\left\{AF^2 +  \frac i2 A^3 F - \frac 1{10} A^5 \right\}
         \ee

    What one cannot do in this case is to impose a self-duality condition. This can be done and the distinguished instanton configurations can be built up only in the space of dimension $D = 4n$.
    
    \section*{Akcnowledgements}
    
    I am indebted to Robert Bryant, Jose Figueroa-O'Farrill, Maxim Kontsevich, Olaf Lechtenfeld,  Nikita Nekrasov, George Savvidy and Ilia Smilga for illuminating discussions.


\begin{thebibliography}{96}

\bibitem{FubNic} E. Corrigan, C. Devchand, D. Fairlie and J. Nuyts, {\it First-order equations for gauge fields in spaces of dimension greater than four}, Nucl. Phys. {\bf B214} (1983) 452;

S. Fubini and H. Nicolai, {\it The octonionic instanton}, Phys. Lett. {\bf B155} (1985) 369.

\bibitem{Keph} B. Grossman, T.W. Kephart and J.D. Stasheff, {\it Solutions to gauge field equations in eight dimensions: Conformal invariance and the last Hopf map}, Commun. Math. Phys. {\bf 96} (1984) 431.

\bibitem{Tchr}
 D.H. Tchrakian, {\it Spherically symmetric gauge field configurations with finite action in 4p dimensions (p = integer)}, Phys. Lett. {\bf B150} (1985) 360.

\bibitem{BPST} A.A. Belavin, A.M. Polyakov, A.S. Schwarz and Y.S. Tyupkin, {\it Pseudoparticle solutions of the Yang-Mills equations}, Phys. Lett. {\bf B59} (1975) 85.

\bibitem{ADHM} M.F. Atiyah, N.J. Hitchin, V.G. Drinfeld and Y.I. Manin, {\it Construction of instantons}, Phys. Lett. {\bf A65} (1978) 185;

E. Corrigan, D.B. Fairlie, S. Templeton and P. Goddard, {\it A Green's function for the general selfdual gauge field}, Nucl. Phys. {\bf B140} (1978) 31.

\bibitem{ADHM-8} A. Nakamula, S. Sasaki and K. Takesue, {\it ADHM construction of (anti-)self-dual instantons in eight dimensions}, Nucl. Phys. {\bf B910} (2016) 199,  arXiv: 1604.01893 [hep-th];

E.K. Loginov, {\it Octonionic instantons in eight dimensions}, Phys. Lett.  {\bf B816} (2021) 136244., arXiv: 2003.09601 [hep-th]. 

\bibitem{string} M.J. Duff and J.X. Lu, {\it Strings from five-branes}, Phys. Rev. Lett. {\bf 66} (1991) 1402;

R. Minasian, S.L. Shatashvili and P. Vanhove, {\it Closed strings from $SO(8)$ Yang-Mills instantons}, Nucl. Phys. {\bf B613} (2001) 87, arXiv: hep-th/0106096.

\bibitem{Hall} B.A. Bernevig, J. Hu, N. Toumbas and S.-C. Zhang, {\it The eight-dimensional quantum Hall effect and the octonions}, Phys. Rev. Lett. {\bf 91} (2003) 236803, arXiv: cond-mat/0306045.

\bibitem{kniga} See e.g. the last equation in [A.V. Smilga, {\it Differential Geometry through Supersymmetric Glasses}, World Scientific, 2020]. 

\bibitem{Mimura} M. Mimura, {\it The homotopy groups of Lie groups of low rank},  J. Math. Kyoto Univ., {\bf 6-2} (1967) 131.

\bibitem{Nikita} N.A. Nekrasov, private communication.

\bibitem{semispin} See e.g. [B. McInnes, {\it The semispin groups in string theory},  J.Math.Phys. {\bf 40} (1999) 4699, arXiv:hep-th/9906059.]

\bibitem{Uhlen} L.L. Uhlenbeck, {\it Removable singularities in Yang-Mills fields}, Commun. Math. Phys. {\bf 82} (1982) 11.

\bibitem{Hooft-Ans} G. `t Hooft, {\it Computation of the quantum effects due to a four-dimensional pseudoparticle}, Phys. Rev. {\bf D14} (1976) 3432.

\bibitem{Zee} B. Zumino, Yong-Shi Wu and A. Zee, {\it Chiral anomalies, higher dimensions and differential geometry}, Nucl. Phys. {\bf B239} (1984) 477.
 
\bibitem{Steen} See e.g.   Theorem 8.6 in [N. Steenrod, {\it Topology of Fibre Bundles}, Princeton Univ. Press, 1999].

\bibitem{OSe} D. O’Se and D. H. Tchrakian, {\it Conformal properties of the BPST instantons of the generalized Yang-Mills system}, Lett. Math. Phys. {\bf 13}  (1987) 211.

\bibitem{Takesue} K. Takesue, {\it ADHM construction of (anti-)self-dual Instantons
in 4n dimensions},  JHEP {\bf 07} (2017) 110, arXiv:1706.03518.


\bibitem{Khanna} G. Khanna, S. Mukhopadhyay, R. Simon and N. Mukunda, {\it Geometric phases for $SU(3)$ representations and three-level quantum systems},  Ann. Phys. {\bf 253} (1997) 55.
 








 

 
 
 
 


\end{thebibliography}
\end{document}